\documentclass[10pt,reqno,tbtags]{amsart}
\usepackage{caption}
\usepackage[latin1]{inputenc}
\usepackage[top=4cm , bottom=4cm, left=4cm , right=4cm]{geometry}
\usepackage{amsmath,amssymb,amsfonts,amsthm} 
\usepackage{mathrsfs,latexsym} 
\usepackage{mathtools}
\usepackage{thmtools}
\numberwithin{equation}{section}
\usepackage{hyperref}
\usepackage{float}
\usepackage{wasysym}
\usepackage{cite}\medskip
\usepackage{todonotes}
\roman{enumi}
\newtheorem{theorem}{Theorem}[section]

\newtheorem{proposition}[theorem]{Proposition}

\numberwithin{equation}{section}

\usepackage[bottom]{footmisc}

\newcommand{\cA}{\mathcal{A}}
\newcommand{\Ecal}{\mathcal{E}}
\newcommand{\Fcal}{\mathcal{F}}
\newcommand{\Scal}{\mathcal{S}}
\newcommand{\Tcal}{\mathcal{T}}
\newcommand{\Ucal}{\mathcal{U}}
\newcommand{\Vcal}{\mathcal{V}}
\newcommand{\BV}{\mathcal{BV}}
\newcommand{\Ci}{\mathcal{C}^{\infty}}
\newcommand{\T}{\cdot_{\mathcal{T}}}

\newcommand{\ext}{\mathrm{ext}}

\newcommand{\ph}{\varphi}

\newcommand{\NN}{\mathbb{N}}

\newcommand{\RR}{\mathbb{R}}

\newcommand{\supp}{\mbox{supp}}

\newcommand{\be}{\begin{equation}\nonumber}
\newcommand{\ee}{\end{equation}}

\def\eg{{\it e.g.\ }}
\def\ie{{\it i.e.\ }}

\newcommand{\const}{\mathrm{const}}

\newcommand{\mc}{\mu c}
\newcommand{\WF}{\mathrm{WF}}
\newcommand{\rpar}[2]{\frac{#1\!\!\stackrel{\leftarrow}{\delta}}{\delta{#2}}}
\newcommand{\lpar}[2]{\frac{\stackrel{\rightarrow}{\delta}\!\!#1}{\delta{#2}}}
\newcommand{\Pei}[2]{\lfloor #1, #2 \rfloor}
\usepackage{wasysym}
\usepackage{hyperref}
\usepackage{cprotect}

\begin{document}
\title{Locally covariant approach to effective quantum gravity}
\author[R. Brunetti]{Romeo Brunetti}
\address{Dipartimento di Matematica, Universit\`a di Trento, 38123 Povo (TN), Italy}
\email{romeo.brunetti@unitn.it}

\author[K. Fredenhagen]{Klaus Fredenhagen}
\address{II. Institute f\"ur Theoretische Physik, Universit\"at Hamburg, 22761 Hamburg, Germany}
\email{klaus.fredenhagen@desy.de}

\author[K. Rejzner]{Kasia Rejzner}
\address{Department of Mathematics, University of York, YO10 5DD York, UK}
\email{kasia.rejzner@york.ac.uk}

\thanks{Invited chapter for the Section ``Perturbative Quantum Gravity'' of the ``Handbook of Quantum Gravity'' (Eds. C. Bambi, L. Modesto andÊ
     I.L. Shapiro, Springer Singapore, expected in 2023)}
\begin{abstract} Despite the fact that quantum gravity is non-renormalisable, a consistent and mathematically rigorous construction of a perturbation series is possible.
    This is based on the use of the Batalin-Vilkovisky-Becchi-Rouet-Stora-Tyutin formalism for gauge theories, the methods of perturbative algebraic quantum field theory and the principle of local covariance. The truncation of the series can be interpreted as an effective quantum field theory which provides predictions for observations at sufficiently small energy scales. Quantum cosmology can be seen as its lowest order expansion, and precision measurements on the cosmic microwave background yield the first empirical test of this approach to quantum gravity. 
\end{abstract}
\maketitle
\section{Introduction}
Fundamental physics at small scales is successfully described by the Standard Model of particle physics \cite{goldberg}. This is a  quantum field theory (QFT) model,
 where the basic objects are quantum fields defined as operator-valued distributions on Minkowski space \cite{SW00}. Effects of external gravitational fields can be included by generalizing the model to Lorentzian spacetimes; provided these spacetimes are globally hyperbolic, a consistent framework has been developed. This is Locally Covariant Quantum Field Theory \cite{BFV03,HW01}. The short distance problems caused by the curvature of the underlying spacetime are nowadays well understood, the interpretation of states, however, is not yet so clear. The reason is that the interpretation of states of QFT on Minkowski space is mainly done in terms  of particles, but there is no known good concept of particles on generic spacetimes, a fact which is visible in some approaches 
 as ``particle creation by curvature'' \cite{parker}. As a consequence, an experimental confirmation of the generalization to QFT on curved backgrounds is not yet available. Nevertheless, there are some predictions of this framework, in particular the radiation of black holes discovered by Hawking \cite{hawking}. The latter prediction seems to be rather convincing \cite{FH90}, but a direct observation for presently known black holes is not possible, since the derived temperature of this radiation is too small. The expected existence of this radiation, however, can be considered as a hint for the way, gravity is connected to quantum physics. Actually, it gave already rise to a rich and somewhat speculative literature (see, \eg \cite{Haw05} and related papers).  
 
 In general, at typical scales of present physics, the quantum effects of the interaction of gravity with other fields seem to be rather small such that their neglection does not lead to conflicts with observation. 
Since gravity is very well described by general relativity which is a classical field theory of the spacetime metric, the direct way of unifying gravity with quantum physics is to add the spacetime metric as an additional field to the fields of the Standard Model. As a classical field theory this is well defined \cite{FR12}. As a quantum field theory, however, severe problems occur.
 
 The main approach to the Standard Model is via perturbation theory. So one first builds a theory of non-interacting fields, where the dynamics is governed by a Lagrangian which is a quadratic functional of the basic fields. The higher order terms of the Lagrangian are then treated as perturbations, and after several decades of hard work, the expansion could be constructed in terms of a formal power series where the lower orders could be explicitly calculated and yield a good, and often excellent agreement with experiments. Crucial for this expansion is that there is only a finite number of parameters (around 25) which have to be determined by experiments so that predictions are possible, a fact which is due to the renormalisability of the model. 
 
 If one tries a corresponding approach to gravity, one has to choose the non-interacting theory, for instance the metric of Minkowski space with a theory of massless particles with helicity 2 (``gravitons''), and uses the difference to the full Einstein Hilbert Lagrangian as an interaction. This interaction is, however, not a polynomial in the gravitational field and its derivatives. Since in 4 spacetime dimensions non-polynomial functionals of the free fields are not well defined (all relatively local fields are Wick polynomials, \cite{borchers,epstein}), one replaces the interaction by a formal power series of polynomials. But then, it turns out that in the construction of the perturbative expansion, in every order new free parameters occur, so that strictly speaking no predictions are possible \cite{tHV,GS86}.
 
 But the latter statement is too pessimistic. Namely, these free parameters have mass dimensions which are increasing with the order of the perturbative expansion. Hence at sufficiently small scales their influence is negligible. The fact that quantum effects of gravity are not easily visible indicates that our present observations  take place at such scales.      

The aim is therefore to ignore the problem of non-renormalisability (in the sense of the occurrence of a finite, nonzero number of free parameters in every order) and to develop a consistent theory of quantum gravity as a formal power series. After calculating the theory up to a given order, one can then  determine the free parameters at this order and check whether the predictions agree with observations \cite{Don22}. 

There are, however, further complications due to the necessity to satisfy the condition of general covariance. So the formalism should start from an arbitrary globally hyperbolic metric, but different choices have to lead to the same theory. Moreover, one has to define local observables, in spite of the fact that diffeomorphisms might change the localization region. Hence one has to build diffeomorphism invariant combinations of the basic fields. In classical field theory this can be achieved by using some of the fields as coordinates and considering the remaining fields as functions of these coordinates (\textit{relational observables}  \cite{BFHPR16}, for a review see \eg \cite{Tamb12}). It is by no means obvious how this can be done with quantum fields as operator valued distributions. Moreover, the causal relations of the underlying spacetime are not stable against variations of the metric, but typically manifest themselves as algebraic relations between quantum fields.

We will describe how a consistent framework for effective quantum gravity coupled to a scalar field can be formulated, and we show that in lowest nontrivial order it just coincides with the models used in quantum cosmology (see, \eg \cite{DS20}) to explain the fluctuations of the cosmic microwave background. This might be seen as a first visible quantum effect of gravity.

\section{Relational observables}\label{sec:inv}
We consider a model of several scalar fields as a simplified version of the Standard Model coupled to the spacetime metric. Then we use scalar fields $X_i$ which may be functions of the elementary fields and such that for generic values of the elementary fields the fields $X_i$ form a coordinate system. In the quantized theory we expand the theory around such a background.

Given a $4$-dimensional manifold  $M$ diffeomorphic to $\RR^4$, we consider configurations $\Gamma=(g,\phi)$ where $g$ is a Lorentz metric such that $(M,g)$ is globally hyperbolic, and $\phi$ is a smooth function with values in $\RR^n$. The configurations are smooth sections of an affine bundle over $M$. Let $\Ecal(M)$ denote the space of configurations $\Gamma$. This is an infinite dimensional affine manifold.

Diffeomorphisms $\chi$ of $M$ act on these sections in a natural way, and this action extends uniquely to an action on the associated jet bundle and induces an action of the space $\Fcal(M)$ of smooth functionals on $\Ecal(M)$. We then consider functions $X_i,i=1,\dots,4$ on the jet bundle which transform as scalar fields, \ie $\chi^{\ast}X_i(x)=X_i(\chi(x)),\ x\in M$. Examples of such functions are the scalar fields $\phi_i$ themselves and also suitable functionals of the metric, for instance the traces of powers of $R_\mu^\nu$, the Ricci tensor multiplied with the inverse metric. We now choose a configuration $\Gamma_0$ such that $x\mapsto X[\Gamma](x)$
is a diffeomorphism from $M$ to $\RR^4$ for $\Gamma$ sufficiently near to $\Gamma_0$. Let $\alpha[\Gamma ]$ denote the diffeomorphism of $M$ such that 
\begin{equation}
    X[\Gamma]\circ\alpha[\Gamma]=X[\Gamma_0]\ .
\end{equation}
Given any tensor field $A[\Gamma]$, we define \cite{BFHPR16}
\be
\mathcal{A}[\Gamma]=\alpha[\Gamma]^{\ast}A[\Gamma]
\ee
where $\alpha[\Gamma]^{\ast}$ is the pullback associated to $\alpha[\Gamma]$,
and observe that $\mathcal{A}$ is invariant under diffeomorphisms, $\mathcal{A}[\chi^{\ast}\Gamma]=\mathcal{A}[\Gamma]$. Integrating $\mathcal{A}[\Gamma]$ with any test function, we obtain a smooth functional, i.e. an element of $\Fcal(M)$.

We now expand $\mathcal{A}[\Gamma]$ around $\Gamma_0$ and obtain a formal power series of fields on the tangent space of the configuration space at the background $\Gamma_0$. Up to first order this is
\begin{equation}\label{eq:inv1}
\mathcal{A}[\Gamma]=A[\Gamma_0]+\left\langle \frac{\delta A}{\delta \Gamma}[\Gamma_0],\Gamma-\Gamma_0\right\rangle -
\mathfrak{L}_ZA[\Gamma_0]
\end{equation}
with the Lie derivative $\mathfrak{L}_Z$ for the vector field $Z=\langle\frac{\delta}{\delta\Gamma}\alpha[\Gamma_0],\Gamma-\Gamma_0\rangle$. We see in particular, that fields which vanish on the background configuration are at first order invariant under diffeomorphisms, a fact which is known as the Stewart-Walker theorem \cite{SW74}.
\section{BV-BRST formalism for gravity}\label{sec:BV}
\subsection{Main ideas of the formulation}
As customary for gauge theories we describe the action of infinitesimal diffeomorphisms by a fermionic vector field $c$ (the ghost field) and  the Becchi-Rouet-Stora-Tyutin (BRST) operator $\gamma$ as the exterior derivative with respect to the action of diffeomorphisms on field configurations. More concretely, the action of $\gamma$ on the metric and on the matter fields $\phi$ is given in terms of the Lie derivative, so 
\[
\gamma g=\mathcal{L}_c g\,,\qquad \gamma \phi=\mathcal{L}_c \phi
\]
where the ghost is treated as the evaluation functional on the space $\mathfrak{X}(M)$ of vector fields on $M$. For a functional $F\in \Fcal(M)$, we have
\[
\gamma F= \left<\frac{\delta F}{\delta g},\mathcal{L}_c g\right>+\left<\frac{\delta F}{\delta \phi},\mathcal{L}_c \phi\right>\,,
\]
and for the ghost itself, 
\[
\gamma c=-\frac{1}{2}[c,c,]\,,
\]
which is now an antisymmetric bilinear functional (bilinear form) on $\mathfrak{X}(M)$.
We require $\gamma$ to satisfy the graded Leibniz rule, which makes it a differential on the space of $n$-forms on  $\mathfrak{X}(M)$, valued in functionals of the metric.  It is convenient to think about antilinear forms on $\mathfrak{X}(M)$ as functionals of fermionic variables $c^{\mu}$. Together with the field configurations $\Gamma$, the ghosts form a graded affine manifold that we denote $\overline{\Ecal}(M)$ (extended configuration space). An element of $\overline{\Ecal}(M)$ is a field multiplet $\ph$, where the components $\varphi^i$ run through all elementary fields of the model, \ie the scalar fields $\phi_j$, the components $g_{\mu\nu}$ of the spacetime metric and the ghosts $c^{\mu}$. We will keep denoting the space of functionals on $\overline{\Ecal}(M)$ by $\Fcal(M)$.

In order to include also the field equations, we extend the space of functionals to the space of \textit{vector fields} $\Vcal(M)$ on the extended configuration space
which can formally be written as 
\[
\int X^i(x) \frac{\delta}{\delta \ph^i(x)}\ .
\]
The functional derivatives $\ph_i^{\ddagger}\equiv\frac{\delta}{\delta \ph^i}$
are called \emph{antifields}
and are treated as densities. Since we are dealing with graded quantities here, we need to specify whether we differentiate from the right or from the left. Unless stated otherwise, all the derivatives are left derivatives.

It turns out to be convenient to embed field configurations, and antifields in a graded manifold, where antifields have the opposite parity as the associated fields. The space of functions on that space is a graded commutative algebra. In this algebra one introduces an odd Poisson bracket $\{.,.\}$, the Schouten bracket, also called antibracket. For a vector field $X\in\Vcal(M)$ and a functional $F\in \Fcal(M)$ the bracket is just the application of the vector field to the  functional,
\be
\{X,F\}=XF\ ,
\ee
and for two vector fields $X,Y\in\Vcal(M)$ it coincides with the Lie bracket. For more general entries, it is extended by the graded Leibniz rule.
This way we obtain a graded Poisson algebra $\mathcal{BV}(M)$, the Batalin-Vilkovisky (BV) algebra. See \cite{BFR16} for more details.

The action $S$ (for now we can think of it formally as a functional on $\Ecal(M)$, the precise formulation follows in the next section) is the sum  of the Einstein-Hilbert action and the action of the scalar fields which we assume to have the form
\be
\int\frac12\sum_j d\phi_j \wedge \ast d\phi_j- V(\phi)\ ,
\ee
with the Hodge dual $\ast$ of the metric and an interaction density $V$. 
The field equations are obtained by the Schouten bracket  of the action with the antifields and give rise to the Koszul-Tate differential
\be
\delta(\bullet)=\{S,\bullet\}.
\ee              
The action is invariant under diffeomorphisms  and hence $\gamma S=0$. Moreover, since $\gamma$ is a graded derivation with respect to the Schouten bracket we have
\be
(\delta\gamma+\gamma\delta)(\bullet)=\{S,\gamma(\bullet)\}+\gamma\{S,\bullet\}=\{\gamma(S),\bullet\}=0.
\ee
Therefore, the \emph{classical BV operator} 
\[
s=\gamma+\delta\,,
\] 
satisfies $s^2=0$. The cohomology of $s$ then yields the gauge invariant on shell observables of the classical theory. 

We now expand $\gamma$ in terms of antifields and obtain
\be
\gamma(\bullet)=\left\{\int\sum_i \gamma(\varphi^i)\varphi_i^{\ddag},\bullet\right\}
\ee       
 We can then write  $s$ as a Schouten bracket with the extended action
\be
S_{\mathrm{ext}}=S+\int\sum_i \gamma(\varphi^i)\varphi_i^{\ddag}\ ,
\ee         
and the extended action satisfies the equation
\begin{equation}\label{eq:CME}
\{S_{\mathrm{ext}},S_{\mathrm{ext}}\}=0\,,
\end{equation}
called \textit{the classical master equation}. We can express the classical BV operator as:
\[
s=\{S_{\ext},\bullet\}\,.
\]

 In order to implement gauge fixing conditions we extend $\overline{\Ecal}(M)$ with additional scalar fields $b_j$, $\bar{c}_j$, $j=0,\dots, 3$. Consequently, $\BV(M)$ gets extended with the associated antifields and one needs to add extra terms to the action $S_{\mathrm{ext}}$. The \emph{Nakanishi-Lautrup fields} $b_j$ are bosonic and transform as scalar fields under $\gamma$,
 \be
 \gamma(b_j)=c^{\nu}\partial_{\nu}b_j\ .
 \ee
 The fields $\bar{c}_j$ (the antighosts) are fermionic and transform as
 \begin{equation}
 \gamma(\bar{c}_j)=-ib_j+c^{\nu}\partial_{\nu}\bar{c}_j \ .
 \end{equation}
 Moreover, $\delta (\bar{c}_j)=\delta (b_j)=0$.
In this way the cohomology of $s$ is not changed. The new extended action is still denoted by $S_{\ext}$.
  
 As gauge fixing fermion we use
 \be
 \Psi=\int\sum_j i\bar{c}_j\left(d\ast dx_j-\tfrac{1}{2}b_jd\mathrm{vol}\right)\,,
 \ee
 with the canonical coordinates $x_j$ of $\RR^4$ and the density $d\mathrm{vol}$ induced by the metric.
 The new extended action is obtained by applying the canonical transformation of the graded Poisson algebra $\mathcal{BV}(M)$, which is induced by $\Psi$, i.e.:
 \[
\ph^i\mapsto \ph^i\,\quad, \ph_i^\ddagger\mapsto \ph_i^\ddagger+\frac{\delta \Psi}{\delta\ph^i}\,.
 \]
 It is given in terms of the old one by
 \be
 S_{\ext}\mapsto S_{\ext}+s(\Psi)=S_{\ext}+\int\sum_j b_j(d\ast d x_j-\frac12 b_jd\mathrm{vol})-i\bar{c}_j d\ast d c^{j} 
 \ee
 up to higher order terms in the antifields.
 
 The theory with the action in this form can be quantized, as the term with no antifields induces normally hyperbolic equations of motion. Different gauge fixings correspond to different canonical transformation of the original graded Poisson algebras and formally (i.e. under infinitesimal changes of the gauge fixing fermion) they are equivalent. In the final step, one sets antifields to zero and the resulting action is called \emph{the gauge-fixed action}.
 \subsection{Precise formulation}\label{sec:precise}
 One of the difficulties one faces when making all the ideas presented above precise is that non-trivial solutions to normally hyperbolic field equations on globally hyperbolic spacetimes cannot be compactly supported. This means that in order to have non-trivial on-shell theory, one cannot restrict oneself to compactly supported configurations. At the same time, globally hyperbolic spacetimes are necessarily non-compact, so the action $S$ cannot be defined simply as the corresponding Lagrangian density $\mathcal{L}$ integrated over the whole $M$, and integration over compact subregions yields singular functionals. Instead, we use the approach of \cite{BDF,BFR}, where a \textit{generalized Lagrangian} is a map $L$ from the space of test functions $\Ci_c(M)$ to the space of local functionals. This map has to be local and in mathematical terms this is expressed by the requirement that it is a natural transformation between certain functors. More concretely, given the usual Lagrangian density $\mathcal{L}(x)[\ph]$ as a map on the jet bundle associated to the extended configuration space, one can define
 \[
L(f)\doteq \int \mathcal{L}[f\ph]\,.
 \]
 Actions are equivalence classes of Lagrangians under the equivalence relation
 \[
L_1\sim L_2\ \mathrm{iff}\ \supp((L_1-L_2)(f))\subset \supp(df)\,,
 \]
 where the \emph{spacetime} support of a functional $F\in \Fcal(M)$ is defined by
 \begin{multline*}
\supp F\doteq \{x\in M|\forall\ \Ucal\ni x \textrm{ open neighborhood}\ \exists \ph,\psi\in \overline{\Ecal}\ \textrm{with}\ \supp(\ph-\psi)\subset\Ucal\\
\textrm{s.t.} \ F(\psi)\neq F(\ph)\}
 \end{multline*}
 The classical master equation \eqref{eq:CME} is weakened to
 \[
\{S_{\mathrm{ext}},S_{\mathrm{ext}}\}\sim 0\,,
 \]
 where $S_{\mathrm{ext}}$ is the equivalence class corresponding to the generalized Lagrangian $L_{\mathrm{ext}}$. We express the classical BV operator as
\[
sF=\{L_{\ext}(f),F\}\,,
\]
where $f\equiv 1$ on supp $F$. We denote this operation by $\{S_{\ext},\bullet\}$.
 
 Classical dynamics are implemented by means of the \emph{Euler-Lagrange derivative}, which is a 1-form on $\overline{\Ecal}(M)$ defined by:
 \[
\left<dL(\ph),\psi\right>\doteq \left< \frac{\delta L(f)}{\delta \ph}(\ph),\psi\right>\,,
 \]
 where $f\equiv 1$ on support $\psi$. Here $\psi$ is compactly supported ($\psi\in\overline{\Ecal}_c(M)$), so $\frac{\delta L(f)}{\delta \ph}(\ph)\in \overline{\Ecal}'_c(M)$ is a distributional density without a restriction on support. For graded field configurations, $\frac{\delta}{\delta \ph^i}$ is the left derivative. Note that $dL$ depends only on $S[L]$, the equivalence class of $L$.

 For the purpose of quantisation, we will need a stronger version of this condition. Implementing this in practice for quantum gravity requires us to use not one, but two different test functions $f_1$, $f_2$, where $f_1$ is used for the Einstein-Hilbert and scalar fields Lagrangians and $f_2$ is used to multiply the gauge field $c$. This yields that the gauge transformations $\gamma$ are compactly supported, and moreover we require that $f_1\equiv 1$ on the support of $f_2$ since in this way the gauge transformations do not see the cut-off of the matter-metric part of the Lagrangian. Eventually, we define
 \[
L_{\textrm{ext}}(f_1,f_2)[\ph,\ph^\ddagger]\doteq\int \mathcal{L}(f_1g,f_1\phi, f_1 \overline{c}, f_1 b, f_2 c;\ph^\ddagger) 
 \]
and the antifields are also transformed, so that $\ph^\ddagger_i\equiv \frac{\delta}{\delta (f_1\ph^i)}$ for all $i$ apart from the ghost indices for which we have ${c^\mu}^\ddagger\doteq \frac{\delta}{\delta (f_2c^\mu)}$. With this definition, the gauge-invariance of the original action implies that 
\[
\{L_{\textrm{ext}}(f_1,f_2),L_{\textrm{ext}}(f_1,f_2)\}=0\,.
\]
\section{Perturbation  around a background}
We choose a background configuration $\Gamma_0=(g_0,\phi_0)$ and $c=b=\bar{c}=0$. 
The background configuration is chosen such that the dynamical coordinates $X_j[\Gamma_0]$ form a coordinate system. The background values of antifields are also set to zero. As a simple example we consider $4$ scalar fields  with $V=0$ and use as background
$\phi_{0_{j}}=x_j$ and a metric solving the Einstein equations with the energy momentum tensor $T_{\mu\nu}(\phi_0)$ such that also the field equation for $\phi_0$ is satisfied. Then we can choose $X_j=\phi_j$ as
dynamical coordinates. 

We expand the generalized Lagrangian $L_{\ext}(g_0+\kappa h,\phi_0+\kappa\varphi,\kappa c,\kappa b,\kappa \bar{c}, \kappa \ph^{\ddagger})$
in $\kappa$ and obtain the decomposition 
\be
L_{\ext}=\kappa^2 L_{0}+L_I(\kappa)+\mathrm{const}
\ee
where $L_0$ contains the terms of second order and $L_I$ starts with terms of at least third order. Both $L_0$ and $L_I$ can be expanded with respect to the antifield number, so that
\[
L_0=L_{00}+L_{01}\,,\qquad L_I=L_{I0}+L_{I1}
\]

For the chosen gauge fixing fermion, the Euler Lagrange equation for $L_{00}$ is Green hyperbolic \cite{Bar},
so the free theory can be constructed by means of deformation quantization and we construct the full theory as a formal power series in $\kappa$ using perturbative algebraic quantum field theory (pAQFT) methods.

The Feynman propagator obtains a factor  $\kappa^{-2}$. The time-ordered powers of the interaction Lagrangian $L_I$ are then formal power series in
$\kappa$.
The algebra of observables is now defined as the cohomology of the \emph{quantum BV operator}, which is a deformation of the classical BV operator $s$.
This contains the diffeomorphism invariant formal power series obtained by expanding fields as functions of the scalar fields $\phi_j$. 

Let us make these ideas more precise. We write the field equation for the generalized Lagrangian $L_{00}$ in the form:
\[
dL_{00}(\ph)=P\ph=0\,,
\]
where $P$ is a Green hyperbolic operator. In terms of the components of the field multiplet $\ph$, we have
\begin{equation}\label{eq:P}
P_{ij}(x)=\lpar{}{\ph^i(x)}L_{00}(f)\rpar{}{\ph^j(x)}\,,
\end{equation}
with left and right derivatives, where $f\equiv 1 $ on a compact neighborhood of $x$. Similarly
\begin{equation}\label{eq:K}
K^{i}_{\ j}(x)=\lpar{}{\ph^\ddagger_j(x)} L_{01}(f) \rpar{}{\ph^i(x)} \ .
\end{equation}

We know that on globally hyperbolic spacetimes there exist retarded and advanced Green functions for $P$. We denote them by $\Delta^{\rm R/A}$. The Poisson bracket of the free theory is introduced using the Pauli-Jordan function
\[
\Delta\doteq \Delta^{\rm R}-\Delta^{\rm A}.
\]
and is defined by
\[
\Pei{F}{G}\doteq \left<\lpar{F}{\ph},\Delta\, \rpar{G}{\ph}\right> \,,
\]
where we suppressed all the indices. Here $F$ and $G$ are smooth functionals with smooth derivatives (i.e. $\frac{\delta F}{\delta \ph}(\ph)$, $\frac{\delta G}{\delta \ph}(\ph)$ are smooth for all $\ph\in\overline{\Ecal}(M)$)

Note that this expression does not involve derivatives with respect to the antifields. To see that the bracket is also well-defined on the $0$-th cohomology of $s_0=\{S_{0},\bullet\}$ (the space of gauge-invariant on-shell observables of the linearised theory), we need to check whether $s_0$ is a graded derivation with respect to the bracket $\Pei{.}{.}$. We expand $s_0$ in the antifield number as $s_0=\delta_0+\gamma_0$, where $\delta_0=\{S_{00},\bullet\}$ implements the linearised equations of motion and $\gamma_0=\{S_{01},\bullet\}$ is the linearised BRST operator.

For $\delta_0$, the argument is clear since $\Delta$ is itself a bi-solution for $P$ and $\delta_0 \ph^{\ddagger}=P\ph$. For $\gamma_0$, we use the fact that the linearized field equations are invariant under the linearized BV operator $\gamma_0$, so
\[
K \Delta^{\mathrm R/A}+\Delta^{\mathrm R/A} K^\dagger=0\,,
\]
where $K^\dagger$ is the formal adjoint of $K$ with respect to the pairing $\left<.,.\right>$. The same goes for $\Delta$. Consequently, $\gamma_0$ is a derivation with respect to $\Pei{.}{.}$.

To deform this resulting Poisson algebra, we need to pick a state for the free (i.e. linearized) theory. We choose a quasifree Hadamard state with a 2-point function $\Delta^+$. Being a Hadamard state means that $\Delta^+$ is of positive type (i.e. $\Delta^+(\bar{f},f)\geq 0$, where $f$ is a test function and $\bar{f}$ is its complex conjugate), the imaginary part of $\Delta^+$ is given by $\frac{i}{2}\Delta$, $\Delta^+$ is a bisolution of $P$ and, crucially,  $\Delta^+$ satisfies the following wavefront set condition:
\[
\mathrm{WF}(\Delta^+)=\{(x,k;x',-k')\in\dot{T}^*M^2|(x,k)\sim(x',k'), k\in (\overline{V}_+)_x\}\,,
\]
where $(x,k)\sim(x',k')$ means that there exists a lightlike geodesic connecting the spacetime points $x$ and $x'$ with $k$ and $k'$ their respective cotangent vectors where the last one is the parallel transport of the former, moreover $\overline{V}_\pm$ is (the dual of) the closed future/past lightcone, seen as a subset of the cotangent bundle and $\dot{T}^*M^2$ is the cotangent bundle deprived of its zero section. The wavefront set characterises the singularity structure of $\Delta^+$ and the above condition essentially says that we want to select states with the same singularity structure  as the Minkowski vacuum. We write
\[
\Delta^+=\frac{i}{2}\Delta+H\,,
\]
where $H$ is the symmetric part, dependent on the choice of the Hadamard state. We define the star product corresponding to this choice as
\[
(F\star G)(\ph)=e^{\hbar\left<\lpar{}{\ph_1},\Delta^+\, \rpar{}{\ph_2}\right>}F(\ph_1)G(\ph_2)\big|_{\ph_1=\ph_2=\ph}\,.
\]
Here $F$ and $G$ are smooth and have derivatives that satisfy wavefront set conditions that make them compatible with the wavefront set of $\Delta^+$. Basically, one can multiply two distributions, if their wavefront sets do not add up to a set which includes a zero section of the cotangent bundle. The precise condition is that
 \begin{equation}\label{mlsc}
\WF(F^{(n)}(\ph,\ph^\ddagger))\subset \Xi_n,\quad\forall n\in\NN,\ \forall\ph\in\overline{\Ecal}(M)\,,
\end{equation}
where $\Xi_n$ is an open cone defined as 
\begin{equation}\label{cone}
\Xi_n\doteq T^*M^n\setminus\{(x_1,\dots,x_n;k_1,\dots,k_n)| (k_1,\dots,k_n)\in (\overline{V}_+^n \cup \overline{V}_-^n)_{(x_1,\dots,x_n)}\}\,.
\end{equation}
The same for $G$. Functionals satisfying this condition are called \textit{microcausal} and we use the notation $\BV_{\mc}(M)$ for this space of functionals. Equipped with the star product $\star$, this is the extended algebra associated to the free theory. For the star product to be compatible with $s_0$, we also need to require that 
\[
K \Delta^{+}+\Delta^{+} K^\dagger=0\,,
\]
i.e. that our Hadamard state is gauge invariant. With this extra requirement, $\star$ is well-defined on the cohomology of $s_0$, which is now the space of gauge-invariant on-shell observables of free theory of the free quantum theory.
Introducing the interaction is done using the methods of perturbative AQFT, as described in \cite{BDF,RejBook,DueBook}. The interaction Lagrangian is $L_I$ and after inserting a test function (or a pair of test functions, as described in section \ref{sec:precise}), we obtain a compactly supported functional $V\equiv L_I(f_1,f_2)$. 

Under these conditions, we can define interacting fields using the Bogoliubov formula. First we need the time-ordered products. Naively, time-ordered products would be defined in terms of the Feynman propagator,
\[
\Delta^F=\frac{i}{2}(\Delta^{\rm R}+\Delta^{\rm A})+H\,,
\]
as
\[
(F\T G)(\ph)=e^{\hbar\left<\lpar{}{\ph_1},\Delta^F\, \rpar{}{\ph_2}\right>}F(\ph_1)G(\ph_2)\big|_{\ph_1=\ph_2=\ph}\,.
\]
This, however, makes sense only provided that $F$ and $G$ are \emph{regular} functionals, i.e. all their derivatives are smooth. Since interesting observables are more singular than that, we need to employ (the extension of the) \emph{Epstein-Glaser renormalisation} \cite{EG,BF00,HW01} to extend the time-ordered product to more general arguments.

One starts with defining the $n$-fold time-ordered products $\Tcal_n(F_1,\dots F_n)$, where the functionals $F_i$, $i=1,\dots n$ are formal power series with coefficients in local functionals. This includes the relational observables  $\cA[\Gamma]$ from section \ref{sec:inv}, which can be expanded in a formal power series of local fields (again denoted by $\cA[\Gamma]$), and smeared with a test density yield a sequence of local functionals, provided that the coordinates depend locally on the fields, 
\[
\cA[\Gamma](f)\doteq\int \cA[\Gamma](x)f(x)\, .
\]

Epstein-Glaser renormalisation allows one to construct each  $\Tcal_n$, assuming that all the lower order $k$-fold time-ordered products ($k<n$) have been constructed and that they fulfill certain natural axioms. The most important axiom, which makes the procedure work, is the \emph{causal factorisation property}
\[
\Tcal_n(F_1\otimes\dots\otimes F_n)=\Tcal_{k}(F_1\otimes\dots\otimes F_k)\star \Tcal_{n-k}(F_{k+1}\otimes\dots\otimes F_n)\,,
\]
 if the spacetime supports of $F_1,\dots,F_k$ do not intersect the past of
the spacetime supports of $F_{k+1},\dots,F_n$ (w.r.t. the background metric).

The Epstein-Glaser (EG) renormalisation is a well-defined procedure if we work with formal power series in $\hbar$ and $\kappa$, even though the theory is power-counting non-renormalisable. However, the extensions obtained using the EG scheme are not unique and this non-uniqueness is described using the St{\"u}ckelberg-Petermann renormalisation group \cite{BDF,PS,SP}. Roughly speaking, renormalisation group transformations  amount to adding finite counterterms, i.e modifying the couplings. The difference between renormalisable and non-renormalisable theories is that in the former case, the total number of such parameters is finite and does not increase as we go to the higher orders in perturbation theory. This is not the case for gravity, which is power-counting non-renormalisable. However, we can still treat it as an effective theory, if we truncate the series at a finite order in $\hbar$ and $\kappa$.

With the $n$-fold time-ordered products at hand, we can define $\T$ as a binary operation using the operator
\[
\Tcal\doteq \bigoplus_{n=0}^\infty \Tcal_n\circ m^{-1}\,,
\]
where $m^{-1}$ is the operation opposite to multiplication that allows one to factorise multilocal functionals into local ones \cite{FR13}. This works also for products of relational observables constructed out of local fields, since they are expressed as power series in local functionals. Using $\Tcal$, we deform the pointwise product $\cdot$ into the renormalised time-ordered product:
\[
F\T G\doteq \Tcal(\Tcal^{-1}F\cdot \Tcal^{-1} G)\,.
\]
One also deforms the classical BV operator of the free theory into
\[
\hat{s}_0\doteq \Tcal^{-1}\circ s_0\circ \Tcal\,,
\]
the quantum BV operator of the free theory. Applying $\Tcal$ can be thought of as normal ordering, so given a classical observable $F$, $\Tcal F$ is a corresponding free quantum observable.  

Interaction is introduced by means of the quantum M{\o}ller maps. Let $F$ be an observable of classical theory and $\Tcal F$ the corresponding observable of the free quantum theory. We define the S-matrix associated with $V$ by
\[
\Scal(V)\doteq e_{\Tcal}^{i\Tcal V/\hbar}=\Tcal e^{iV/\hbar}\,.
\]
Using the Bogoliubov formula, we can now write down the quantum M{\o}ller operator $R_V$ that allows us to construct interacting quatum observable corresponding to $F$:
\[
R_V(F)\doteq \Scal(\Tcal V)^{-1}\star (\Scal(\Tcal V)\T \Tcal F)=(\Tcal e^{i V/\hbar})^{-1}\star \Tcal(e^{iV/\hbar}\cdot F)\,.
\]
As in the free case, we can deform the classical linearised BV operator to obtain the interacting quantum BV operator:
\[
\hat{s}\doteq R_V^{-1}\circ s_0 \circ R_V\,.
\]
This operator is a local operator, if we assume the quantum master equation (QME), which can be expressed as the condition:
\[
s_0\Scal(V)=0\,.
\]
This is equivalent to
\[
\frac{1}{2}\{L_{\textrm{ext}}(f_1,f_2),L_{\textrm{ext}}(f_1,f_2)\}-i\hbar \triangle_V=0\,,
\]
where $ \triangle_V$ is the anomaly term that can be calculated using the anomalous Master Ward identity \cite{BreDue,FR13}. As explained in \cite{BFR16} we can use the renormalisation freedom in defining the time-ordered products to ensure that the QME is fulfilled. Then the interacting quantum BV operator takes the form:
\[
\hat{s}F=sF-i\hbar \triangle_VF\,,
\]
where $\triangle_VF\doteq \frac{d}{d\mu} \triangle_{V+\mu F}\big|_{\mu=0}$ is the renormalised \emph{BV Laplacian}. The cohomology of $\hat{s}$ describes the quantum gauge-invariant on-shell observables.

It remains to discuss the dependence of the chosen background and gauge fixing. The background independence can be shown by generalizing the formalism to backgrounds which are not solutions of the field equations. One then can show that a change of the background within a compact region induces a trivial automorphism of the algebra of gauge invariant on shell observables \cite{BFR16}. The independence of the choice of gauge fixing is intrinsic to the BV formalism, since a change of the gauge fixing fermion amounts to a symplectic transformation of the BV algebra which does not change the cohomology of the BV operator. Nevertheless, it would be necessary to show that this property is not lost by renormalization.  (See the discussion in \cite{TehraniZahn}.) 

\section{Cosmological perturbation theory}
We apply the general formalism to the case of gravity minimally coupled to a real scalar field (the dilaton) with self interaction $V$ \cite{BFHPR16}. As a background we choose a Friedmann-Lemaitre-Robertson-Walker (FLRW) spacetime with metric $g=a^2\eta$
where $a$ is a function of the conformal time $\tau\equiv x^0$, 
\be
\eta\equiv \eta_{\mu\nu}d x^{\mu}d x^{\nu}=-d\tau^2+\sum dx_i^2\ ,
\ee
and as the background for the dilaton $\phi$, we choose a function $\phi_0$ depending only on $\tau$. The equations of motions are satisfied, if $a$ and $\phi_0$ fulfill the equations
\begin{equation}
    \begin{split}
(\phi_0')^2+2a^2V(\phi_0)&=6\mathcal{H}^2\\
-(\phi_0')^2+2a^2V(\phi_0)&=2(2\mathcal{H}'+\mathcal{H}^2)\\
        \phi_0''+2\mathcal{H}\phi_0'+a^2\frac{dV}{d\phi}(\phi_0)&=0
    \end{split}
\end{equation}
with $\mathcal{H}=a'/a$ and where $\bullet'$ denotes the derivative with respect to $\tau$.
$\mathcal{H}$ is related to the Hubble parameter $H$ by $\mathcal{H}=Ha$ and to the Ricci scalar $R$ by $R=6(\mathcal{H}'+\mathcal{H}^2)a^{-2}$. Note that the third equation follows from the first two if $\phi_0'$ nowhere vanishes. We assume from now on that $\phi_0'<0$.

A problem with this background is that due to its high symmetry there are not sufficiently many local functionals of the fields which can serve as coordinates. For one coordinate function we use the dilaton field $\phi$ itself and set 
\begin{equation}\label{eq:dilatontime}
X_0=\phi_0^{-1}\circ\phi\equiv T\ 
\end{equation}
as our dynamical time function. The other coordinates are constructed as follows.
We restrict the spacetime metric to the  hypersurfaces $\phi=\const$. For $\phi$ near to $\phi_0$ they are spacelike. 
We use the spatial coordinates $x_i$ for the hypersurface given by $\phi=\mathrm{const}$ and we will also need the following definition of the tangent fields at the hypersurface,
\begin{equation}\label{tangent}
    \partial_{i,\phi}=\partial_i-\frac{\partial_i\phi}{\partial_\tau \phi}\partial_\tau\ .
\end{equation}
The components of the induced metric with respect to these coordinates are
\begin{equation}
    g_{\phi, ij}\doteq g_{\phi}(\partial_{i,\phi},\partial_{j,\phi})=g_{ij}-g_{i0}\frac{\partial_j\phi}{\partial_\tau\phi}-g_{j0}\frac{\partial_i\phi}{\partial_\tau\phi}+g_{00}\frac{\partial_i\phi}{\partial_\tau\phi}\frac{\partial_j\phi}{\partial_\tau\phi}\ .
\end{equation}
Let $\triangle_{\phi}$ denote the Laplacian with respect to the induced metric on this hypersurface (not to confuse with the graded BV Laplacian from the previous section). Explicitely, it is given by:
\begin{equation}\label{eq:lap}
\triangle_{\phi}=\frac{1}{\sqrt{\det g_\phi}}\partial_{i,\phi}g_{\phi}^{ij}\sqrt{\det g_{\phi}}\partial_{j,\phi}=\frac{\partial_{i,\phi}\sqrt{\det g_{\phi}}}{\sqrt{\det g_{\phi}}}g_{\phi}^{ij}\partial_{j,\phi}+\partial_{i,\phi}g_{\phi}^{ij}\partial_{j,\phi}\ .
\end{equation}
We then define
\begin{equation}
    X_i=(1-G_{\phi}\triangle_{\phi})x_i
\end{equation}
where $G_{\phi}$ is the Green operator for $\triangle_{\phi}$ with vanishing boundary conditions at infinity.

The diffeomorphism $\alpha_{\Gamma}$ then assumes the following form.
\begin{proposition}
Let $\Gamma$ be a compactly supported variation of the background\\ (\ie $\mathrm{supp}(\Gamma-\Gamma_0)$ is compact). Then 
\be
\alpha_{\Gamma}(\tau,x)=(\tau+\delta\tau,x+\delta x)
\ee
with
\be
\delta\tau=\phi(\bullet,x+\delta x)^{-1}\circ\phi_0(\tau)-\tau
\ee
and
\be
\delta x_i=\frac{a^2}{4\pi }\int d^3y \frac{\triangle_{\phi}y^i}{|x-y|}\ .
\ee
\end{proposition}
\begin{proof}
On smooth functions $f\in\mathcal{C}^\infty(\RR^3)$ we have 
\be
(1-G_{\phi}(\triangle_{\phi}-\triangle_{\phi_0}))(1+G_{\phi_0}(\triangle_{\phi}-\triangle_{\phi_0}))f=f
\ee
where we use the resolvent equation
\be
G_{\phi}-G_{\phi_0}=G_{\phi}(\triangle_{\phi_0}-\triangle_{\phi})G_{\phi_0}\ .
\ee
which holds on compactly supported smooth functions.

The claim for $\delta x_i$ now follows from the fact that $\triangle_{\phi_0}$ annihilates the coordinate functions and from the explicit form of $G_{\phi_0}$. The claim for $\delta\tau$ is a simple consequence of the definition of $T$.
\end{proof}
The proposition is expected to hold also for perturbations $\Gamma-\Gamma_0$ which are not compactly supported but vanish sufficiently fast at infinity. However,
we refrain from entering this issue in the present review.

Unfortunately, the chosen spatial coordinates are non-local.
As a consequence, the  associated invariant fields  are invariant only under diffeomorphisms which tend sufficiently fast to the identity at infinity. This creates problems in higher order perturbation theory where renormalisation ambiguities for non-local functionals are not under control, in general.
Up to first order perturbation theory, however, this problem does not appear.

In zeroth order, we have a free theory which coincides with the traditional cosmological perturbation theory \cite{DS20}, where the linearised classical equations hold also
in the quantum theory. Quantum effects arise from the commutation relations and from the correlations in appropriate states.

The commutation relations are uniquely fixed by $L_{00}$. It is given by
\begin{equation}
    L_{00}(h,\varphi,c,b,\bar{c})=\frac12\langle (h,\varphi),\tilde{P}(h,\varphi)\rangle+\langle(h,\varphi),Qb\rangle-\frac12\sum_j ( b_j^2 d\mathrm{vol}-i\bar{c_j} d\ast d c^j)
\end{equation}
$\tilde{P}$ and $Q$ are differential operators depending on the background configuration $\Gamma_0$. (See \cite{H14} for details, where $P$ and $\tilde{P}$ are interchanged and $Q$ is denoted by $K$). The operator
\begin{equation}
    P=\left(\begin{array}{cc}
        \tilde{P} & Q \\
        Q^t & -1
    \end{array}\right)
\end{equation}
is Green-hyperbolic \cite{Bar}, since $\tilde{P}+QQ^t$ is normally hyperbolic. The advanced Green operator $\Delta^{\rm A}$ for $P$ is obtained from the advanced Green operator
$E_{A}$ of $\tilde{P}+QQ^t$ by 
\begin{equation}
    \Delta^{A}=\left(\begin{array}{cc}
       E_{A}  & E_AQ\\
       Q^tE_A  &Q^tE_AQ 
    \end{array}\right)\ ,
\end{equation}
and an analogous formula holds for the retarded Green operator.

The ghost term of the Lagrangian is decoupled at this order, hence it can be treated separately.

Let now $\Delta$ be the difference between the retarded and the advanced Green operator of $P$. The smeared fields
\be
A(f)=h(f_1)+\varphi(f_2)+b(f_3)
\ee
satisfy the commutation relations 
\be
[A(f),A(g)]=i\langle f,\Delta g\rangle
\ee
and the field equation
\be
A({P}^tf)=0\ .
\ee
Taking the generators $A(f)$ and the relations above, we obtain an algebra.

Another way to obtain this algebra, up to isomorphy, is to consider the algebra generated by the linear functionals of the form $A(f)$ and apply the procedure described in section~\ref{sec:BV}. This amounts to taking the 0-homology of the Koszul-Tate operator, which then corresponds to taking the quotient by expressions of the form $A({P}^tf)$.

The ghosts and antighosts satisfy canonical anticommutation relations with the scalar commutator function multiplied by $-i$ and the corresponding field equation.

The linearized BRST transformation $\gamma_0$  maps linear fields to linear fields. It is given by
\be
\gamma_0(h_{\mu\nu})=\nabla_\mu c_{\nu}+\nabla_\nu c_\mu
\ee
with the Levi-Civita connection $\nabla$ of the background metric $g_0$,
\be
\gamma_0(\varphi)=c^\mu\partial_\mu\phi_0\ ,\ \gamma_0(\bar{c}_j)=-ib_j
\ee
and vanishes on the other fields.

The fields which are at first order gauge invariant are obtained by the expansion described in Section \ref{sec:inv}, equation \eqref{eq:inv1}. In the case at hand, the vector field $Z$ is given by
\be
Z=\frac{\varphi}{\phi_0'}\partial_\tau+Z^i \partial_i
\ee
with
\begin{equation}\label{eq:Zi}
Z^i=\left\langle \frac{\delta X_i}{\delta \Gamma}[\Gamma_0],\Gamma-\Gamma_0\right\rangle=G_0\left(\partial_i\frac{\mathcal{H}}{\phi_0'}\varphi+\partial_j h_i^j-\frac12\partial_i h_j^j\right)\equiv X_i^{(1)}\,,
\end{equation}
where $G_0$ is the Green operator of the Laplacian on $\RR^3$ (i.e. the convolution with the Coulomb potential). To obtain the above formula, we used the fact that differentiating \eqref{eq:lap}, one obtains
\be
\left\langle\frac{\delta\triangle_{\phi}}{\delta\phi}[a^2\eta,\phi_0],\varphi\right\rangle=-\frac{1}{a^{2}\phi_0'}\left(2(\partial_i\varphi)\partial_i\partial_\tau+\mathcal{H}(\partial_i\varphi)\partial_i+(\triangle\varphi)\partial_{\tau}\right)\,.
\ee
and
\be
\left\langle \frac{\delta\triangle_{\phi}}{\delta g}[a^2\eta,\phi_0],h\right\rangle=\sum_{ij}a^{-4}\left(\frac12(\partial_ih_{jj})\partial_i-\partial_i h_{ij}\partial_j\right)\ .
\ee

Next, we find that 
\begin{equation}\label{eq:gammazero}
\gamma_0(Z^{\mu})=c^{\mu}\,.
\end{equation}
Now we use formula \eqref{eq:inv1}, with $A$ the metric and where $A[\Gamma_0]$ is $a^2\eta$ and obtain the gauge invariant fields
\begin{equation}
    \tilde{g}_{\mu\nu}=a^2\eta_{\mu\nu}+h_{\mu\nu}-\partial_{\mu}Z_\nu-\partial_\nu Z_\mu+2\Gamma_{\mu\nu}{}^{\lambda}Z_\lambda\ .
\end{equation}
Using \eqref{eq:gammazero} we verify that $\gamma_0(\tilde{g}_{\mu\nu})=0$.

The non-vanishing Christoffel symbols for the background metric $a^2\eta$ are
\be
\Gamma_{00}{}^0=\Gamma_{ii}{}^0=\Gamma_{i0}{}^i=\Gamma_{0i}{}^i=\mathcal{H}\ ,\ i=1,2,3 .
\ee
Defining $\tilde{h}=\tilde{g}-a^2\eta$ and using the convention
$Z_{\mu}=a^2\eta_{\mu\nu}Z^{\nu}$, we obtain
\be
\tilde{h}_{00}=h_{00}-2\partial_0Z_0+2\mathcal{H}Z_0=h_{00}+2a^2(\partial_{\tau} +\mathcal{H})\frac{\varphi}{\phi_0'}
\ee
\be
\tilde{h}_{0i}=h_{0i}-\partial_0Z_i-\partial_i Z_0+ 2\mathcal{H}Z_i=h_{0i}-a^2\partial_{\tau}X_i^{(1)}+a^2\partial_i\frac{\varphi}{\phi_0'}
\ee
\be
\tilde{h}_{ij}=h_{ij}-\partial_iZ_j-\partial_jZ_i+2\mathcal{H}\delta_{ij}Z_0=h_{ij}-a^2(\partial_iX_j^{(1)}+\partial_jX_i^{(1)}+2\mathcal{H}\delta_{ij}\frac{\varphi}{\phi_0'})\ .
\ee
All the fields other than $\tilde{h}_{00}$ are non-local due
to the occurrence of $G_0$. However, local fields can be easily obtained by applying the Laplace operator to them. Since the dilaton field was used as a coordinate, the corresponding gauge invariant field
\be
\tilde{\phi}=\phi_0+\varphi-Z^0\partial_\tau\phi_0=\phi_0
\ee
is trivial. Using the formula \eqref{eq:inv1} we also find that  $\tilde{X}_i^{(1)}=0$. Now applying \eqref{eq:Zi} we obtain  the following constraints
\be
\sum_j\partial_j\tilde{h}_{ji}=\sum_j\frac12\partial_i\tilde{h}_{jj}\,,
\ee
so the fields $\tilde{h}_{ij}$ are not independent.

It is customary to parametrise the metric perturbation $h$ in terms of scalars, vector and tensor fields with respect to the euclidean symmetry of $\RR^3$. Note however that this parametrisation is unique only if all these fields vanish at infinity. Moreover, these fields are non-local functionals of the field configuration. 

The space-space-components of $h$ in that parametrisation are given by
\be
h_{ij}=2a^2(\partial_i\partial_jE+\delta_{ij}D+\partial_{(i}W_{j)}+T_{ij})
\ee
with a tensor field $T$ with $\sum_iT_{ii}=0$ and $\sum_i\partial_iT_{ij}=0$, a vector field $W$ with $\sum_i\partial_iW_i=0$ and scalar fields $D$ and $E$ with $D=\frac12(\sum_jh_{jj}-\sum_{ij}G_{0}\partial_i\partial_jh_{ij})$ and $E=G_0(\frac{1}{2a^2}\sum_j h_{jj}-3D)$. In terms of these fields, the 1st order coordinate fields\footnote{Unfortunately, in \cite{BFHPR16}, the field $W$ was missing in the corresponding formula} are
\be
X_i^{(1)}=\partial_iE+W_i+\frac{\mathcal{H}}{\phi_0'}G_0\partial_i\mu
\ee
with the Mukhanov-Sasaki variable $\mu=\varphi-\frac{\mathcal{H}}{\phi_0'}D$.

For the gauge invariant fields $\tilde{h}_{ij}$ we find
\be
\tilde{h}_{ij}=2a^2\bigl(T_{ij}+\frac{\mathcal{H}}{\phi_0'}(G_0\partial_i\partial_j-\delta_{ij})\mu\bigr)\ .
\ee
The time-time component is $h_{00}=-2a^2 A$ with a scalar field $A$. The corresponding gauge invariant field is
\be
\tilde{h}_{00}=2a^2\bigl((\partial_\tau+\mathcal{H})\frac{\varphi}{\phi_0'}-A\bigr)\ . 
\ee
The time-space component is written as
\be
h_{0i}=a^2(V_i-\partial_i B)
\ee
with a scalar field $B$ and a divergence free vector field $V$. For the gauge invariant vector field $\tilde{h}_{0\bullet}$ we get
\be
\tilde{h}_{0i}=a^2\bigl((V-W')_i+\frac{1}{\phi_0'}\partial_i(\chi-G_0\mu')\bigr)
\ee
with the vector field $V-W'$ and the scalar field $\chi=\varphi-\phi_0'(B+E')$.

We observe that our gauge invariant fields can be parametrised by the fields $\mu$, $\chi$, $T$, $V-W'$ and $\Phi=A-(\partial_{\tau}+\mathcal{H})(B+E')$, which are the gauge invariant fields  traditionally used in cosmological perturbation theory. Moreover, these fields are uniquely determined by the fields $\tilde{h}_{\mu\nu}$. Note, however, that the mentioned parametrisation induces additional non-localities, which can give rise to spurious violations of causality \cite{Eltz}. On the contrary, the fields $\triangle\tilde{h}_{\mu\nu}$ are local and have to satisfy the usual causal relations. 

This construction of gauge invariant fields may be illustrated by relating it to geometrical objects. We use the tangent fields in \eqref{tangent} and
compute the spatial curvature tensor at first order
\be
\begin{split}
    {R^{\phi\,\, l}_{ijk}}=\frac{1}{2a^2}&(\partial_i\partial_k h_{jl}-\partial_j\partial_k h_{il}-\partial_i\partial_l h_{jk}
    +\partial_j\partial_l h_{ik})\\
    &+\frac{\mathcal{H}}{a\phi_0'}(\partial_i\partial_l\varphi\delta_{jk}-\partial_j\partial_l\varphi\delta_{ik}-\partial_i\partial_k\varphi\delta_{jl}+\partial_j\partial_k\varphi\delta_{il})\\
    =\frac{1}{2a^2}&(\partial_i\partial_k \tilde{h}_{jl}-\partial_j\partial_k \tilde{h}_{il}-\partial_i\partial_l \tilde{h}_{jk}+\partial_j\partial_l \tilde{h}_{ik}) \ .
\end{split}
\ee
It is already gauge invariant, in agreement with the Stewart-Walker theorem, since it vanishes on the background.

The extrinsic curvature, however,
\be
K_{ij}^{\phi}=-N_{\phi}\langle d\phi,\nabla_{\partial_{i,\phi}}\partial_{j,\phi}\rangle, \text{ with }N_{\phi}=|g^{-1}(d\phi,d\phi)|^{-\frac12}\ ,
\ee
 assumes on the background the value
\be
K_0^{\phi}{}_{ij}=a\mathcal{H}\delta_{ij}
\ee
and therefore its first order contribution
\be
{K_1}^{\phi}_{ij}=\frac{\mathcal{H}}{2a}h_{00}\delta_{ij}-\frac{a}{\phi_0'}\partial_i\partial_j\varphi
\ee
is not gauge invariant.

A gauge invariant 1st order contribution $\tilde{K}^{\phi}_1$ to the extrinsic curvature is obtained by evaluating it at an infinitesimally shifted conformal time $\tau'(\varphi,\tau)$, which amounts to replace $h_{00}$ by $\tilde{h}_{00}$ and $\varphi$ by 0, i.e.
\be
\tilde{K}_1^{\phi}{}_{ij}=\frac{\mathcal{H}}{2a}\tilde{h}_{00}\delta_{ij}\ .
\ee
\section{Quantum gravity and the cosmic microwave background}
The observed cosmic microwave background (CMB) can, to a large extent, be described by a state of the free electromagnetic field on an FLRW spacetime which satisfies the Kubo-Martin-Schwinger (KMS) condition \cite{HHW} with respect to conformal time. This is a quasifree state with the 2-point function
\be
\omega_{\beta}(F_{\mu\nu}(x)F_{\rho\sigma}(y))=(2\pi)^{-3}\int d^{4}p\delta(p^2)
P_{\mu\nu\rho\sigma}(p)e^{i(x-y)p}\frac{\theta(p^0)+e^{\beta p}\theta(-p^0)}{1-e^{-\beta p}}
\ee
with
\be
P_{\mu\nu\rho\sigma}=p_\mu p_\rho\eta_{\nu\sigma}-p_\nu p_\rho\eta_{\mu\sigma}+p_\nu p_\sigma\eta_{\mu\rho}-p_\mu p_\sigma\eta_{\nu\rho}\ .
\ee
$\beta$ is here a 4-vector in the interior of the future light cone in Minkowski space, and we use the notation $yp=y^\mu p_{\mu}$, $p_{\mu}=\eta_{\mu\nu}p^{\nu}$ and $p^2=-p_\mu p^{\mu}$. 
$\beta$ characterizes a rest system determined by the heat bath. Its proper time
\be
|\beta|=\sqrt{\beta^2}
\ee
is the inverse of the temperature relative to conformal time. In cosmic time $t$, related to conformal time $\tau$ by $dt=a d\tau$, this state is not an equilibrium state. But since the scale parameter $a$ varies slowly, one can interpret it approximately as an equilibrium state with time-dependent temperature $\mathfrak{T}=\frac{1}{a|\beta|}$.
Its temperature at the time of recombination is related to the binding energy of hydrogen. This then yields the nowadays observed temperature. 

The electromagnetic radiation in the relevant frequencies does interact only weakly with the mostly neutral matter after recombination. The main deviations from the simple behaviour as a conformal KMS state are due to variations of the metric. This is known as the Sachs-Wolfe effect \cite{SW67}. Some of these variations of the metric are caused by inhomogenities in the distribution of matter in the universe. But in addition one observes small fluctuations which can be interpreted as quantum fluctuations of the gravity-dilaton system described in the previous section. 

The observed electromagnetic field can be related to the field at the time $\tau_r$ of recombination by using the free Maxwell equations in a spacetime with metric $a^2\eta+h$. For this purpose we choose a smooth function $\chi$ of conformal time which is equal to 1 for $\tau<\tau_r-\epsilon$ and vanishes for $\tau>\tau_r+\epsilon$.
The free Maxwell equation for the field strength $F$ (considered as a 2-form)
is
\be
(d+\delta)F=0
\ee
with the codifferential $\delta$.
The differential operator $d+\delta$ is Green hyperbolic \cite{Bar}, \ie it has unique  retarded and advanced Green operators $G_{\mathrm{ret}}$ and $G_{\mathrm{adv}}$. Let now $f$ be a compactly supported 2-form with support contained in a neighborhood of our own spacetime position with $\tau>\tau_r+\epsilon$ for $(\tau,x)\in\supp f$.  
We have
\begin{equation}
    f=(d+\delta)\chi G_{\mathrm{adv}}f+(d+\delta)(1-\chi) G_{\mathrm{adv}}f\ .
\end{equation}
The first term on the right hand side has support in the time slice $\tau_r-\epsilon<\tau<\tau_r+\epsilon$.
Since $G_{\mathrm{adv}}f$ has support in the past light cone of $\supp f$, $(1-\chi) G_{\mathrm{adv}}f$ has compact support. Due to the field equation,
$\int F\wedge(d+\delta)(1-\chi)G_{\mathrm{adv}}f=0$, hence the second term does not contribute to the smeared field strength $\int F\wedge f$, and we obtain
the identity
\be
\int F\wedge f=\int F\wedge (d+\delta)\chi G_{\mathrm{adv}}f
\ee
which relates the present electromagnetic field with the field at the recombination time. 
The Maxwell equations are conformally invariant, therefore we may use the codifferential of the conformally transformed metric $\eta+a^{-2}h$.
For the unperturbed spacetime we then can use the codifferential $\delta_0$ and the Green operator $G_{\mathrm{adv}}^0$ of Minkowski space and get in first order 
for the perturbed metric  
$\delta=\delta_0+\delta_1$ and 
\be
G_{\mathrm{adv}}=G_{\mathrm{adv}}^0-G_{\mathrm{adv}}^0\delta_1G_{\mathrm{adv}}^0\ .
\ee
The measurement of electromagnetic observables related to the radiation from a certain direction then provides information on 
the metric. 
The observed correlations between variations of the metric in different directions can now be related to correlation functions  of the associated observables in the gravity-dilaton system.

A rough estimate of the effect can be obtained from the lapse function (cf. \eqref{eq:dilatontime})
\be
N=|g^{-1}(dT,dT)|^{-\frac12}\ .
\ee
Up to first order it is given by
\be
N=a(1-\frac{\varphi'}{\phi_0'}-\frac{h_{00}}{2a^2})\ .
\ee
A gauge invariant version is
\be
\tilde{N}=a(1-\frac{\tilde{h}_{00}}{2a^2})\ .
\ee
The 2-point function of $\tilde{h}_{00}$ in an appropriate state of the linearized gravity-dilaton system is then related to the observed variations of the temperature of the CMB 
\[\frac{\delta\mathfrak{T}}{\mathfrak{T}}\approx \frac{\tilde{h}_{00}}{2a^2}\ .\]


The partial explanation of the observed temperature fluctuations is a strong support for the existence of the dilaton field and represents the up to now only observed effect of the quantization of gravity.
\section{Concluding remarks}
We have seen that quantum gravity, in spite of its non-renormalisability, gives rise to a well defined perturbation series which can be understood as an effective field theory. The problem of the absence of local observables can be treated in terms of relative observables.
On generic backgrounds, these relative observables are local functionals of dynamical fields which are used as coordinates. On backgrounds with a high symmetry one has to rely on non-local expressions which will create problems in higher order perturbation theory. 

We here described a choice of coordinates which are appropriate for FLRW spacetimes. On Minkowski space as a background, one could instead use solutions of the wave equation
\be
X^{\mu}=(1-G_{\mathrm{ret}}\Square_g)x^{\mu}\ .
\ee
with the canonical coordinates of $\RR^{1,3}$.
This was used in \cite{FRV22} for a computation of the quantum correction to the Newton formula for the gravitational attraction. This choice, which is related to the harmonic gauge, was also used in several papers of Markus Fr\"ob et al. for more general backgrounds \cite{Frob}. It avoids some of the pathologies caused by the spatial non-locality of the spatial coordinates introduced in \cite{BFHPR16}.  

In reality, we have all the fields of the Standard Model which could serve as coordinates, provided we expand around a generic background. So the use of relative observables which are local functionals of local fields is always possible and will yield a well defined perturbation expansion. If we instead use a highly symmetric background for our convenience we have to pay for this with the treatment of non-local quantities.


\begin{thebibliography}{99}
\bibitem{Bar}
C.~B\"ar, ``Green-Hyperbolic Operators on Globally Hyperbolic Spacetimes,''
Commun. Math. Phys. \textbf{333} (2015) 1585-1615.
\bibitem{borchers}
H.~-J.~Borchers, ``\"Uber die Mannigfaltigkeit der interpolierenden Felder zu einer kausalen S-Matrix,'' Il Nuovo Cimento \textbf{15} (1960) 784-794.
\bibitem{BreDue}
F.~Brennecke and M.~D\"utsch, ``Removal of violations of the Master Ward Identity in perturbative QFT,'' Rev. Math. Phys. \textbf{20} (2008), 119-172.
\bibitem{BF00}
R.~Brunetti and K.~Fredenhagen, ``Microlocal Analysis and
Interacting Quantum Field Theories: Renormalization on Physical Backgrounds,'' Commun. Math. Phys. \textbf{208} (2000) 623-661.
\bibitem{BFV03}
R.~Brunetti, K.~Fredenhagen and R.~Verch, 
``The Generally Covariant Locality Principle - A New Paradigm for Local Quantum Field Theory,'' 
Commun. Math. Phys.,\textbf{237} (2003) 31-68.
\bibitem{BDF}
R.~Brunetti, M.~D\"utsch and K. Fredenhagen, ``Perturbative Algebraic Quantum Field Theory and the Renormalization Groups,''
Adv. Theor. Math. Phys.  \textbf{13} (2009) 1541-1599.
\bibitem{BFR16}
R.~Brunetti, K.~Fredenhagen and K.~Rejzner,
``Quantum gravity from the point of view of locally covariant quantum field theory,''
Commun. Math. Phys. \textbf{345} (2016) no.3, 741-779
\bibitem{BFHPR16}R.~Brunetti, K.~Fredenhagen, T.~P.~Hack, N.~Pinamonti and K.~Rejzner,
``Cosmological perturbation theory and quantum gravity,''
J. High Energ. Phys. \textbf{2016}, 32 (2016).
\bibitem{BFR}
R.~Brunetti, K.~Fredenhagen and P.~L.~Ribeiro, ``Algebraic Structure of Classical Field Theory: Kinematics and Linearized Dynamics for Real Scalar Fields,''
Commun. Math. Phys. \textbf{368} (2019) 519-584.
\bibitem{DS20}
S.~Dodelson and F.~Schmidt, \textit{Modern Cosmology}, 2nd ed., Elsevier, 2020.
\bibitem{Don22}
J.~F.~Donoghue, ``Quantum general relativity and effective field theory,'' 	arXiv:2211.09902 [hep-th].
\bibitem{DueBook}
M.~D\"utsch, \textit{From classical field theory to perturbative quantum field theory}, Prog. Math. Phys., vol. 74. Birkh\"auser (2019).
\bibitem{Eltz}
B.~Eltzner, ``Quantization of Perturbations in Inflation,'' arXiv:1302.5358.
\bibitem{epstein}
H.~Epstein, ``On The Borchers Class of a Free Field,'' Il Nuovo Cimento \textbf{27} (1963) 886-893. 
\bibitem{EG}
H.~Epstein and Y.~Glaser, ``The role of locality in perturbation theory,'' Ann. Inst. Henri Poincar\'e-Section
A, vol. XIX, n.3, 211 (1973).
\bibitem{FH90}
K.~Fredenhagen and R.~Haag, ``On the derivation of Hawking radiation associated with the formation of a black hole,''
Commun. Math. Phys. \textbf{127} (1990) 273-284.
\bibitem{FR12}
K.~Fredenhagen and K.~Rejzner, ``Batalin-Vilkovisky Formalism in the Functional Approach to Classical Field Theory,''
 Commun. Math. Phys. \textbf{314} (2012) 93-127.
\bibitem{FR13}
K.~Fredenhagen and K.~Rejzner,
``Batalin-Vilkovisky formalism in perturbative algebraic quantum field theory,''
Commun. Math. Phys. \textbf{317} (2013), 697-725
\bibitem{Frob}
 M.~B.~Fr\"ob, ``Gauge-invariant quantum gravitational corrections to correlation functions,'' Class. Quant. Grav. \textbf{35} (2018) 055006.
\bibitem{FRV22}
M.~B.~Fr\"ob, C.~Rein and R.~Verch, ``Graviton correction to the Newtonian potential using invariant observables,'' J. High Energ. Phys. \textbf{2022}, 180 (2022).
\bibitem{goldberg}
D.~Goldberg, \textit{The Standard model in a nutshell}, Princeton University Press, 2017.
\bibitem{GS86}
M.~H.~Goroff and A.~Sagnotti, ``The Ultraviolet Behavior of Einstein Gravity,''   Nuclear Physics B \textbf{226} (1986) 709-736.  
\bibitem{HHW} 
R.~Haag, N.~M.~Hugenholtz and M.~Winnink, ``On the equilibrium states in quantum statistical mechanics,'' Commun. Math. Phys. \textbf{5} (3) (1967) 215-236.
\bibitem{H14}
T.~P.~Hack,
``Quantization of the linearized Einstein-Klein-Gordon system on arbitrary backgrounds and the special case of perturbations in inflation,''
Class. Quant. Grav. \textbf{31} (2014) no.21, 215004
\bibitem{hawking}
S.~W.~Hawking, ``Particle creation by black holes,'' Commun. in Math. Phys. \textbf{43} (1975) 199-220.
\bibitem{Haw05}
S.~W.~Hawking, ``Information loss in black holes,'' Phys. Rev. D \textbf{72}, 084013 (2005).
\bibitem{HW01}
S.~Hollands and R.~Wald, ``Existence of Local Covariant Time Ordered Products of Quantum Fields in Curved Spacetime,'' 
Commun. Math. Phys., \textbf{231} (2001) 309-345.
\bibitem{tHV}
G.~Õt~Hooft and M.~J.~G.~Veltman, ``One loop divergencies in the theory of
gravitation,'' Ann. Inst. H. Poincare Phys. Theor. A \textbf{20} (1974) 69-94.
\bibitem{parker}
 L.~Parker,``Particle creation in expanding universes.'' Physical Review Letters, \textbf{21} (1968) 562.
\bibitem{PS}
G.~Popineau and R.~Stora, ``A pedagogical remark on the main theorem of perturbative renormalization theory,'' Nuclear Physics B \textbf{912} (2016) 70-78.
\bibitem{RejBook}
K.~Rejzner, \textit{Perturbative Algebraic Quantum Field Theory: An Introduction For Mathematicians}, Mathematical Physics Studies, Springer Cham (2016).
\bibitem{SW67} 
R.~K.~Sachs and A.~M.~Wolfe, ``Perturbations of a Cosmological Model and Angular Variations of the Microwave Background,'' Astrophysical Journal, \textbf{147} (1967) 73.
\bibitem{SW74}
J.M. Stewart and M. Walker,``Perturbations of spacetimes in general relativity,'' Proc. Roy. Soc. Lond. A \textbf{341} (1974) 49.
\bibitem{SW00}
R.~F.~Streater and A.~S.~Wightman, \textit{PCT, Spin and Statistics, and All That}, Princeton Landmarks in Physics, 2000.
\bibitem{SP}
E.~C.~G.~St\"uckelberg and A.~Petermann, ``La normalisation des constantes dans la th\'eorie des quanta,'' Helv. Phys. Acta \textbf{26} (1953) 499-520.
\bibitem{Tamb12}
J.~Tambornino, ``Relational Observables in Gravity: a Review,'' SIGMA \textbf{8} (2012) 017.
\bibitem{TehraniZahn}
M.~T.~Tehrani and J.~Zahn, ``Background independence in gauge theories,'' Ann. Henri Poincar\'e \textbf{21} (2020) 1135-1190.
\end{thebibliography}
\end{document}